\documentclass[conference,a4paper]{IEEEtran}
\IEEEoverridecommandlockouts
\usepackage{amsmath}
\usepackage{amssymb}
\usepackage{amsfonts}
\usepackage{algorithm}
\usepackage{dsfont}
\usepackage{float}
\usepackage{cite}
\usepackage{graphicx}
\usepackage{epsfig}
\usepackage{subfigure}
\usepackage{psfrag}
\usepackage{xcolor}
\usepackage{url}
\usepackage[colorlinks,linkcolor=black,urlcolor=black,anchorcolor=black,citecolor=black,hyperfootnotes=true]{hyperref}
\usepackage{bm}

\def\BibTeX{{\rm B\kern-.05em{\sc i\kern-.025em b}\kern-.08em
    T\kern-.1667em\lower.7ex\hbox{E}\kern-.125emX}}

\newtheorem{remark}{Remark}
\newtheorem{theorem}{Theorem}
\newtheorem{proof}{Proof}

\begin{document}

\title{
		A New Channel Estimation Strategy in Intelligent Reflecting Surface Assisted Networks
	}
	\author{\IEEEauthorblockN{Rui Wang${}^\dagger{}$, Liang Liu${}^\ddagger$, Shuowen Zhang${}^\ddagger$, and Changyuan Yu${}^\ddagger$}
		\IEEEauthorblockA{${}^\dagger$ Shenzhen Research Institute, The Hong Kong Polytechnic University. Email: wangrrui2016@gmail.com \\ ${}^\ddagger$ EIE Department, The Hong Kong Polytechnic University. Email: $\{$liang-eie.liu,shuowen.zhang,changyuan.yu$\}$@polyu.edu.hk}
	}
	\maketitle

\begin{abstract}
Channel estimation is the main hurdle to reaping the benefits promised by the intelligent reflecting surface (IRS), due to its absence of ability to transmit/receive pilot signals as well as the huge number of channel coefficients associated with its reflecting elements. Recently, a breakthrough was made in reducing the channel estimation overhead by revealing that the IRS-BS (base station) channels are common in the cascaded user-IRS-BS channels of all the users, and if the cascaded channel of one typical user is estimated, the other users' cascaded channels can be estimated very quickly based on their correlation with the typical user's channel \cite{b5}. One limitation of this strategy, however, is the waste of user energy, because many users need to keep silent when the typical user's channel is estimated. In this paper, we reveal another correlation hidden in the cascaded user-IRS-BS channels by observing that the user-IRS channel is common in all the cascaded channels from users to each BS antenna as well. Building upon this finding, we propose a novel two-phase channel estimation protocol in the uplink communication. Specifically, in Phase I, the correlation coefficients between the channels of a typical BS antenna and those of the other antennas are estimated; while in Phase II, the cascaded channel of the typical antenna is estimated. In particular, all the users can transmit throughput Phase I and Phase II. Under this strategy, it is theoretically shown that the minimum number of time instants required for perfect channel estimation is the same as that of the aforementioned strategy in the ideal case without BS noise. Then, in the case with BS noise, we show by simulation that the channel estimation error of our proposed scheme is significantly reduced thanks to the full exploitation of the user energy.
\end{abstract}

\begin{figure}
	\begin{center}
		\subfigure[Correlation among different users' channels found in \cite{b5}]
		{\scalebox{0.5}{\includegraphics*{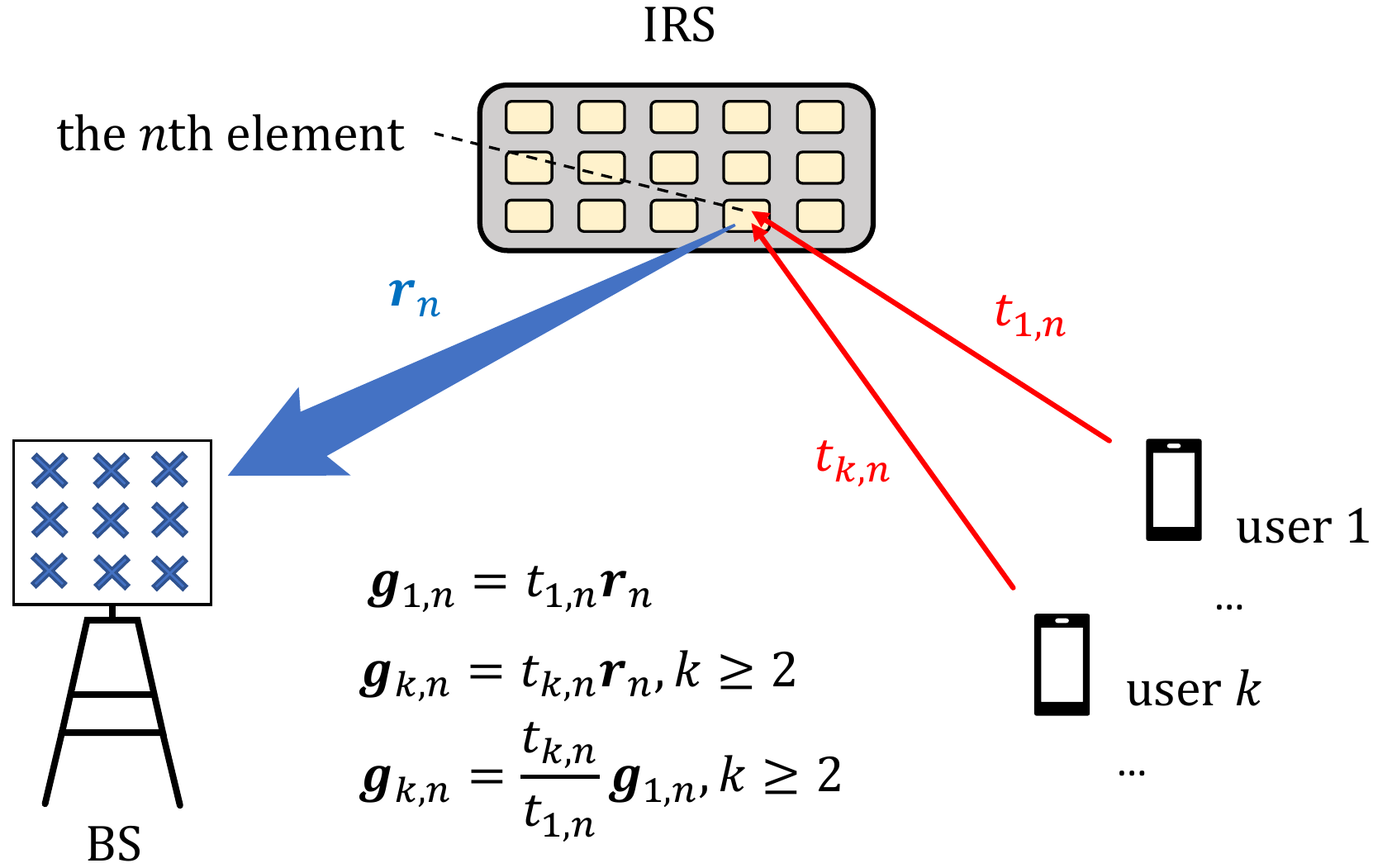}}}
		\subfigure[Correlation among different antennas' channels found in this paper]
		{\scalebox{0.45}{\includegraphics*{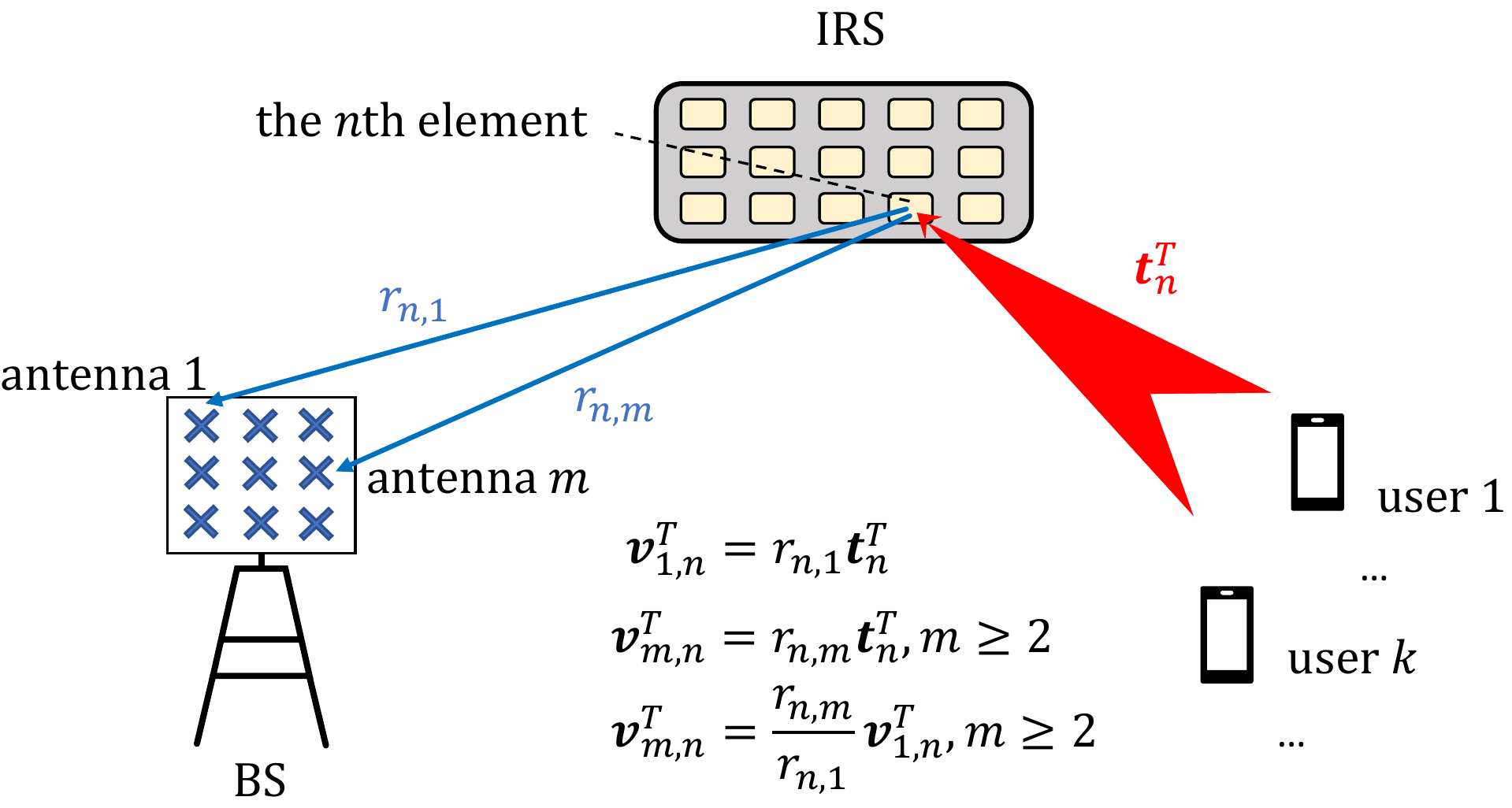}}}
	\end{center}\vspace{-10pt}
	\caption{Various correlations in the cascaded user-IRS-BS channels.}\label{system_model}\vspace{-10pt}
\end{figure}
\section{Introduction}
Recently, there is a large body of research in investigating the fundamental limits of intelligent reflecting surface (IRS) assisted wireless networks \cite{Wu18,Huang19,Zhang19,Schober19,Shuowen20,Pan20,Guo,Yang19}, where each IRS may modify the wireless channels between the base station (BS) and users to be more favorable for communication via inducing phase shift to the incident signal at each reflecting element \cite{Liaskos08,Renzo19,Zhang21}. However, one challenge to approach these limits in practice lies in channel estimation. First, the passive IRS cannot transmit/receive pilot signals actively, making it difficult to estimate the cascaded user-IRS-BS channels, which are products of the user-IRS channels and the IRS-BS channels. Second, the overhead for estimating the cascaded user-IRS-BS channels in each coherence block scales with the number of IRS elements, which is very large in practice.

Recently, \cite{b5} revealed that in a multi-user multi-antenna system, there is a great amount of redundancy in the cascaded user-IRS-BS channels. Specifically, as shown in Fig. \ref{system_model} (a), all the users share the same IRS-BS channel components, and the cascaded channel vector of a user is thus a scaled version of that of any other user. To exploit this correlation to reduce the redundancy in channels, \cite{b5} proposed a novel channel estimation strategy, where the cascaded channel of a typical user is estimated first, and the scaling coefficients, rather than the whole channel vectors, of the other users are estimated next. It was shown that the channel estimation overhead is significantly reduced under this novel scheme.

However, one issue of this scheme is that all the other users cannot transmit their pilot signals when the BS is estimating the typical user's cascaded channel. Since each user has its own power budget, a lot of power is wasted under this scheme, leading to higher channel estimation mean-squared error (MSE), which motivates us to tackle this issue in this paper. Specifically, we consider the uplink communication of an IRS-assisted network, where a multi-antenna BS serves multiple single-antenna users. Our key observation is an alternative correlation hidden in cascaded user-IRS-BS channels as shown in Fig. \ref{system_model} (b): all the BS antennas share the same user-IRS channel components, and the cascaded channel vector of an antenna is thus a scaled version of that of any other antenna. Based on this finding, we propose a new protocol with two phases to estimate the cascaded user-IRS-BS channels. In Phase I, the scaling coefficients, rather than the whole channel vectors, of all the antennas other than a typical antenna are estimated; then, in Phase II, the cascaded channel of the typical antenna is estimated. Note that all the users can transmit throughout Phases I and II under our protocol, because we estimate the cascaded channels antenna by antenna, rather than user by user as in \cite{b5}. Interestingly, it is theoretically shown that in the case without BS noise, the minimum number of time instants required for perfect channel estimation achieved by our strategy is the same as that achieved by the strategy in \cite{b5}. Moreover, in the case with BS noise, the linear minimum mean-squared error (LMMSE) estimators are proposed to reduce the channel estimation MSE. It is shown via simulation that in a practical scenario with BS noise, the channel estimation error of our strategy is much lower than that of the strategy in \cite{b5}, because all the users' power is fully exploited.

The rest of this paper is organized as follows. Section \ref{sec:System Model} introduces the system model. Section \ref{sec:Review of the Channel Training and Estimation Protocol} presents the channel estimation strategy proposed in \cite{b5} and its main limitation. Section \ref{sec:A New Channel Training and Estimation Protocol} proposes a new channel estimation strategy to overcome the limitation in the above strategy. Sections \ref{sec:Minimum Duration for Perfect Channel Estimation under the Case without BS Noise} and \ref{sec:LMMSE Channel Estimation Under the Case with BS Noise} characterize the minimum time duration for perfect channel estimation in the ideal case without BS noise and the LMMSE channel estimators in the practical case with BS noise, respectively. Section \ref{sec:Numerical Results} presents the numerical results. Finally, Section \ref{sec:Conclusions} concludes this paper.

\section{System Model}\label{sec:System Model}
We consider the uplink communication in an IRS-assisted multi-user communication network with $K>1$ single-antenna users and a BS equipped with $M>1$ antennas. An IRS equipped with $N>1$ passive reflecting elements is deployed for enhancing the users' communication performance, as shown in Fig. \ref{system_model}. For all the channels (i.e., user-BS channels, user-IRS channels, and IRS-BS channel), we consider a quasi-static block fading model, where the channels remain approximately constant in each coherence block. Specifically, we denote the direct channel from the $k$th user to the BS as ${\bm h}_{k}\in\mathbb{C}^{M\times 1}, k=1,\cdots,K$. Moreover, we let ${\bm r}_n=\left[ r_{n,1},\cdots, r_{n,M}\right]^T\in\mathbb{C}^{M\times 1}$ and $t_{k,n}\in\mathbb{C}$ denote the channel from the $n$th IRS element to the BS and that from the $k$th user to the $n$th IRS element, respectively,  $n=1,\cdots,N$, $k=1,\cdots,K$. The cascaded user-IRS-BS channel from the $k$th user to the BS via the $n$th IRS element can be denoted as ${\bm g}_{k,n}=t_{k,n}{\bm r}_n\in \mathbb{C}^{M\times 1}, \forall n,k$.

To design the IRS passive beamforming for improving the communication performance, accurate channel state information (CSI) of both the direct user-BS channels ${\bm h}_k$'s and the cascaded user-IRS-BS channels ${\bm g}_{k,n}$'s is needed \cite{b5}. During the channel training phase, we let $\sqrt{p}x_{i,k}\in \mathbb{C}$ denote the pilot signal sent from user $k$ at time instant $i$, where either $|x_{i,k}|^2=1$ or $x_{i,k}=0$ (no pilot symbol) and $p$ denotes the pilot signal power. Moreover, the overall pilot signals of all the users at time instant $i$ is denoted as ${\bm x}_i=[x_{i,1},\cdots,x_{i,K}]^T$. Let $\phi_{n,i}\in \mathbb{C}$ with $|\phi_{n,i}|=1$ denote the reflection coefficient at IRS element $n$ at time instant $i$. The signal received by the BS at time instant $i$ is the superposition of the signals from the users' direct channels and the reflected ones via the IRS, which is expressed as
\begin{align}\label{revsig}
  \bar{{\bm y}}_i  =\sum_{k=1}^{K}{\bm h}_{k}\sqrt{p}x_{i,k}+\sum_{k=1}^{K}\sum_{n=1}^{N}\phi_{n,i}{\bm g}_{k,n}\sqrt{p}x_{i,k}+{\bm z}_{i},
\end{align}
where ${\bm z}_i \sim\mathcal{CN}\left({\bm 0},\sigma^2{\bm I}_M\right)$ denotes the circularly symmetric complex Gaussian (CSCG) noise of the BS at time instant $i$.

Note that the direct user-BS channels ${\bm h}_k$'s can be easily obtained via conventional channel estimation techniques by turning off all the IRS elements. As a result, this paper assumes that ${\bm h}_k$'s have been estimated perfectly such that we can focus on how to tackle the challenges in estimating the cascaded user-IRS-BS channels. Note that with perfect knowledge of ${\bm h}_k$'s, the signal contributed by the direct link in (\ref{revsig}) can be removed, leading to the following signal model useful for estimating the cascaded user-IRS-BS channels:

\begin{align}\label{revsig4}
  {\bm y}_i  =\sum_{k=1}^{K}\sum_{n=1}^{N}\phi_{n,i}{\bm g}_{k,n}\sqrt{p}x_{i,k}+{\bm z}_{i}, ~ \forall i.
\end{align}

It is worth noting that a novel scheme for estimating the reflected channel ${\bm g}_{k,n}$'s was proposed in \cite{b5}, by exploiting the correlation among the user-IRS-BS channels of different users. In the following, we will first review the scheme proposed in \cite{b5} and identify its limitation. Motivated by that, we will introduce an alternative channel training and estimation scheme that achieves higher estimation accuracy by exploiting a newly found correlation in the user-IRS-BS channels.

\section{Review of Channel Estimation Protocol in \cite{b5}}\label{sec:Review of the Channel Training and Estimation Protocol}

In \cite{b5}, a useful correlation was revealed among the user-IRS-BS channels of different users. Specifically, as illustrated in Fig. \ref{system_model} (a), since ${\bm g}_{1,n}=t_{1,n}{\bm r}_n$ and ${\bm g}_{k,n}=t_{k,n}{\bm r}_n$, $\forall k\geq 2$, the user-IRS-BS channel of user $k\geq 2$ can be expressed as a scaled version of that of user $1$ (denoted as the typical user):
\begin{equation}\label{corr}
  {\bm g}_{k,n}=\lambda_{k,n}{\bm g}_{1,n}, ~~k=2,\cdots,K, n=1,\cdots,N,
\end{equation}
where $\lambda_{k,n}$ denotes the correlation coefficient and is given as $\lambda_{k,n}=\frac{t_{k,n}}{t_{k,1}}$. This relation holds because from the user's perspective, the IRS-BS channels ${\bm r}_n$'s are common for various users. Based on this correlation, the $MNK$ coefficients in ${\bm g}_{k,n}$'s can be fully characterized by the user-IRS-BS channels of the typical user, i.e., ${\bm g}_{1,n}$'s,  and the correlation coefficients of the remaining users, i.e., $\lambda_{k,n}$'s, $k=2,...,K$. Therefore, only $MN+(K-1)N$ coefficients need to be estimated, which is much lower than $MNK$. To exploit this correlation, \cite{b5} proposed a novel protocol. First, only the typical user, i.e., user $1$, transmits its pilot signal to the BS such that ${\bm g}_{1,n}$'s can be estimated without the interference from other users' pilots. After estimating ${\bm g}_{1,n}$'s, user $2$ to user $K$ will transmit their pilots to the BS such that $\lambda_{k,n}$'s can be estimated. It was shown in \cite{b5} that in the ideal case without noise at the BS, the minimum number of time instants to perfectly estimate ${\bm g}_{1,n}$'s and $\lambda_{k,n}$'s is
\begin{equation}\label{minoverload}
  \tau_{\rm min}=N+{\rm max}(K-1,\lceil(K-1)N/M\rceil).
\end{equation}

Despite the pioneering contributions made in \cite{b5}, one limitation of the proposed scheme is that only the typical user transmits its pilot in the phase of estimating ${\bm g}_{1,n}$'s. Note that the transmit power of each user $k\geq 2$ in Phase I cannot be saved to increase its transmit power in Phase II, because each user has a peak power constraint in practice. Moreover, the power of the other users cannot be transferred to the typical user in Phase I. As a result, the transmit power available for user $2$ to user $K$ is wasted in Phase I. To harness the benefits brought by exploiting channel correlation yet overcoming the aforementioned limitation, in this paper, we reveal a new type of correlation among the user-IRS-BS channels, based on which we propose a novel channel estimation protocol that allows simultaneous transmission of the uplink pilots from all users throughout the channel training phase.

\section{A New Channel Estimation Protocol}\label{sec:A New Channel Training and Estimation Protocol}

In this section, we present a new channel estimation scheme that can achieve the same minimum channel training duration for perfect estimation given in (\ref{minoverload}) in the case without noise at the BS and potentially reduce the estimation MSE in the case with noise at the BS. Before introducing the estimation scheme, we reveal another correlation relationship, i.e., the correlation between the reflected channels from all the users to different antennas. Note that the reflected channel from all the users to the $m$th antenna via the $n$th IRS element can be expressed as \begin{align}\label{eqn:antenna}
{\bm v}_{m,n}^T=r_{n,m}{\bm t}_{n}^T\in\mathbb{C}^{1\times K}, ~ \forall m,n
\end{align}where ${\bm t}_{n}^T=\left[ t_{1,n},\cdots,t_{K,n}\right]$ denotes the channel from all the users to IRS element $n$. As a result, the relation between ${\bm v}_{m,n}$'s and ${\bm g}_{k,n}$'s can be expressed as
\begin{align}\label{eqn:re}
{\bm v}_{m,n}^T=[g_{1,n,m}, \cdots, g_{K,n,m}], ~ \forall m,n,
\end{align}where $g_{k,n,m}$ denotes the $m$th element in ${\bm g}_{k,n}$. It is observed that a common vector ${\bm t}_n^T$ exists in ${\bm v}_{m,n}^T$'s, $\forall m$, because the user-IRS channel is the same for different antennas, as illustrated in Fig. \ref{system_model} (b). Based on this interesting correlation, each reflected channel vector ${\bm v}_{m,n}^T$ for antenna $m\geq 2$ can be expressed as a scaled version of ${\bm v}_{1,n}^T$ for antenna 1 (denoted as the typical antenna):
\begin{equation}\label{antecorr2}
  {\bm v}_{m,n}^T=\beta_{m,n}{\bm v}_{1,n}^T, ~ m=2,\cdots,M, n=1,\cdots,N,
\end{equation}
where
\begin{equation}\label{beta}
  \beta_{m,n}=\frac{r_{n,m}}{r_{n,1}}, ~ m=2,\cdots,M, n=1,\cdots,N.
\end{equation}

Based on \eqref{eqn:re} and \eqref{antecorr2}, the signal received by the BS at time instant $i$ for estimating the reflected channels shown in (\ref{revsig4}) reduces to:
\begin{align}\label{revsig2}
  {\bm y}_i & =\sqrt{p}\sum_{n=1}^{N}\phi_{n,i}[{\bm v}_{1,n},
    \beta_{2,n}{\bm v}_{1,n}, \cdots,
    \beta_{M,n}{\bm v}_{1,n}]^T{\bm x}_{i}+{\bm z}_{i} \notag \\
  & = \sqrt{p}\sum_{n=1}^{N}\phi_{n,i}{\bm \beta}_{n}{\bm v}_{1,n}^T{\bm x}_{i}+{\bm z}_i,
\end{align}
where ${\bm \beta}_{n}=\left[ 1,{\beta}_{2,n},\cdots, {\beta}_{M,n}\right]^T$. Therefore, our objective is to estimate antenna $1$'s reflected channel vectors ${\bm v}_{1,n}$'s and the correlation coefficients $\beta_{m,n}$'s for antenna 2 to antenna $M$ based on ${\bm y}_i$'s. The number of coefficients to be estimated is reduced from $KMN$ to $KN+(M-1)N$ thanks to the use of the correlation shown in (\ref{antecorr2}).

\begin{figure}[t]
  \centering
  \includegraphics[width=7.8cm]{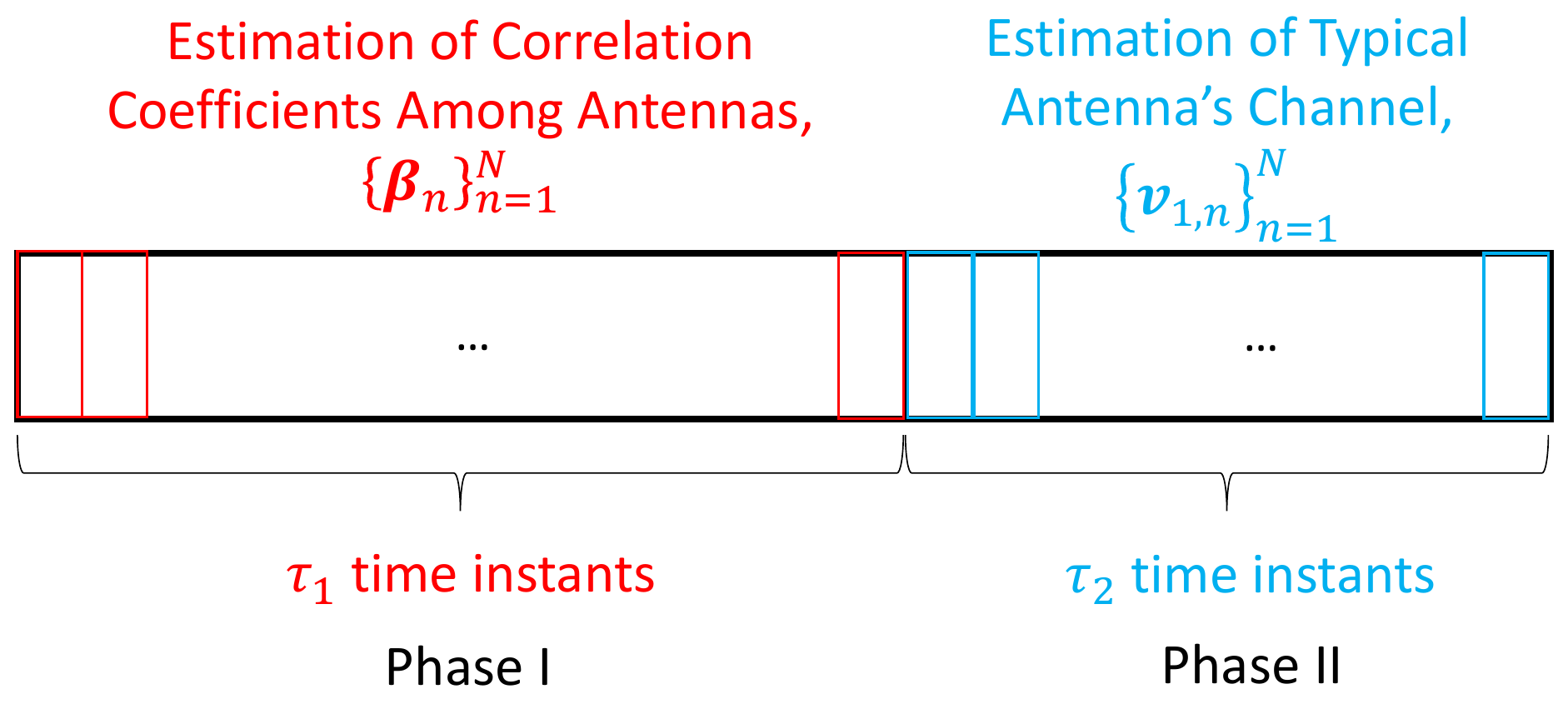}\vspace{-10pt}
  \caption{Illustration of the new channel estimation protocol.}\label{flow_graph}\vspace{-10pt}
\end{figure}

By exploiting the new correlation, we propose a two-phase channel training and estimation protocol as shown in Fig. \ref{flow_graph}. Specifically, in Phase I with $\tau_1$ time instants, the correlation coefficients among different BS antennas $\beta_{m,n}$'s are estimated; in Phase II with $\tau_2$ time instants, the channel vectors for the typical antenna ${\bm v}_{1,n}$'s are then estimated. In the rest of this paper, we will first study this protocol in the theoretical case without noise at the BS to show the minimum number of time instants for perfect channel estimation, and then turn to the practical case with noise at the BS to show how to design the LMMSE channel estimators for reducing the estimation MSE.

\section{Minimum Duration for Perfect Channel Estimation under the Case without BS Noise}\label{sec:Minimum Duration for Perfect Channel Estimation under the Case without BS Noise}
In this section, we characterize the performance limit of our proposed protocol in the case without BS noise.

\subsection{Phase I: Estimation of Correlation Coefficients of Different Antennas}
In this phase, we aim to estimate the correlation coefficients, i.e., $\beta_{m,n}$, $m=2,...,M$ and $n=1,...,N$. According to (\ref{revsig2}), in the case without noise, the received signal of antenna $m$ at time instant $i$ in Phase I can be re-written as
\begin{equation}\label{revS1m}
  \tilde{y}_{m,i}^{\rm I}=\sum_{n=1}^{N}\phi_{n,i}\alpha_{m,n,i}, ~ m=1,\cdots,M, i=1\cdots,\tau_1,
\end{equation}where
\begin{equation}\label{alpha_m}
  \alpha_{m,n,i}=
  \begin{cases}
    $$ \sqrt{p}{\bm v}_{1,n}^T{\bm x}_i$$, & \mbox{$m=1$}, \\
    $$\sqrt{p}\beta_{m,n}{\bm v}_{1,n}^T{\bm x}_i$$, & \mbox{$m\geq 2$},
  \end{cases} ~ \forall n, i.
\end{equation}Note that if $\alpha_{m,n,i}$'s can be estimated based on $\tilde{y}_{m,i}^{\rm I}$'s, the correlation coefficient for antenna $m\geq 2$ and IRS element $n$ can be obtained via
\begin{align}\label{alpha}
  \beta_{m,n}=\frac{\alpha_{m,n,i}}{\alpha_{1,n,i}}, ~ i=1,\cdots,\tau_1.
\end{align}

The estimation of $\alpha_{m,n,i}$'s requires to recover $MN\tau_1$ coefficients. Note that according to (\ref{alpha}), we do not need to estimate all these coefficients, since given any $m$ and $n$, one particular pair of $\alpha_{m,n,\bar{i}}$ and $\alpha_{1,n,\bar{i}}$ for some $\bar{i}$ is sufficient to estimate $\beta_{m,n}$. In other words, estimating $MN$ coefficients among all the $MN\tau_1$ coefficients in $\alpha_{m,n,i}$'s is sufficient. The challenge is given $\tilde{y}_{m,i}^{\rm I}$'s, how to just estimate a subset of $\alpha_{m,n,i}$'s to reduce the training time $\tau_1$.

To tackle this challenge, we propose to set identical pilot signals over all time instants of Phase I, i.e.,
\begin{equation}\label{x_i=x}
  {\bm x}_i={\bm x}, ~ i=1,\cdots,\tau_1.
\end{equation}In this case, the received signal given in (\ref{revS1m}) reduces to
\begin{equation}\label{revS1m1}
  \tilde{y}_{m,i}^{\rm I}=\sum_{n=1}^{N}\phi_{n,i}\bar{\alpha}_{m,n}, ~ m=1,\cdots,M, i=1\cdots,\tau_1,
\end{equation}where
\begin{equation}\label{alpha_m1}
  \bar{\alpha}_{m,n}=
  \begin{cases}
    $$ \sqrt{p}{\bm v}_{1,n}^T{\bm x}$$, & \mbox{$m=1$}, \\
    $$\sqrt{p}\beta_{m,n}{\bm v}_{1,n}^T{\bm x}$$, & \mbox{$m\geq 2$},
  \end{cases} ~ \forall n.
\end{equation}Then, we can first estimate $\bar{\alpha}_{m,n}$'s with $MN$ coeffcients based on the received signals (without considering how to reduce the redundancy in $\alpha_{m,n,i}$'s) and then set
\begin{align}\label{alpha1}
  \beta_{m,n}=\frac{\bar{\alpha}_{m,n}}{\bar{\alpha}_{1,n}}, ~ \forall m,n.
\end{align}

Given (\ref{x_i=x}), the overall received signal at antenna $m$ over $\tau_1$ time instants of Phase I is given by
\begin{align}\label{overallrevS1m}
  \tilde{\bm y}_m^{\rm I} = \left[ \tilde{y}_{m,1}^{\rm I},\cdots,\tilde{y}_{m,\tau_1}^{\rm I}\right]^T
   =  {\bm \Phi}^{\rm I}\bar{{\bm \alpha}}_{m}, ~ \forall m,
\end{align}
where
\begin{equation}\label{Phi1}
  {\bm \Phi}^{\rm I}=\left[
  \begin{array}{ccc}
    \phi_{1,1} & \cdots & \phi_{N,1} \\
    \vdots & \ddots & \vdots \\
    \phi_{1,\tau_1} & \cdots & \phi_{N,\tau_1}
  \end{array}
  \right],
\end{equation}and $\bar{{\bm \alpha}}_{m}=\left[ \bar{\alpha}_{m,1},\cdots, \bar{\alpha}_{m,N}\right]^T$. Then, the received signal at the BS in Phase I is given by
\begin{equation}\label{overallrevS1}
  \tilde{\bm Y}^{\rm I} =\left[ \tilde{\bm y}_{1}^{\rm I},\cdots,\tilde{\bm y}_{M}^{\rm I}\right]= {\bm \Phi}^{\rm I}\left[ \bar{{\bm \alpha}}_{1},\cdots, \bar{{\bm \alpha}}_{M}\right].
\end{equation}Since the numbers of equations and variables are $M\tau_1$ and $MN$, respectively, the minimum number of time instants to perfectly estimate $\bar{\alpha}_{m,n}$'s and thus $\beta_{m,n}$'s based on (\ref{alpha1}) is
\begin{align}\label{eqn:tau1}
\tau_1^\ast=N.
\end{align}In this case, we can set ${\bm \Phi}^{\rm I}$ based on the discrete Fourier transform (DFT) matrix via adjusting the IRS reflecting coefficients such that $({\bm \Phi}^{\rm I})^H{\bm \Phi}^{\rm I}=N{\bm I}$ and thus $[ \bar{{\bm \alpha}}_{1},\cdots, \bar{{\bm \alpha}}_{M}]=({\bm \Phi}^{\rm I})^H\tilde{\bm Y}^{\rm I}/N$.

\begin{remark}
Note that under the protocol proposed in \cite{b5}, in Phase I, user 2 to user $K$ do not transmit their pilot signals to the BS such that the BS can estimate the typical user's channels ${\bm g}_{1,n}$'s without interference. However, under our proposed protocol, it is observed from (\ref{x_i=x}) that all the users can transmit their pilot signals to the BS. Later, it will be shown in Section \ref{sec:Numerical Results} that such a full utilization of the user transmit power in Phase I will significantly reduce the channel estimation MSE compared to the scheme proposed in \cite{b5} in the case with BS noise.
\end{remark}

\subsection{Phase II: Estimation of Reflected Channels for the Typical Antenna}
In the second phase, we aim to estimate the channel vectors $\{{\bm {v}}_{1,n}\}_{n=1}^N$ for the typical antenna, i.e., antenna 1. Before introducing the scheme in Phase II, we want to emphasize that after $\bar{\alpha}_{m,n}$'s are estimated in Phase I, we already have some useful information about ${\bm {v}}_{1,n}$'s as follows:
\begin{align}\label{eqn:v}
\sqrt{p}{\bm v}_{1,n}^T{\bm x}=\bar{\alpha}_{1,n}, ~ \forall n.
\end{align}This information should be used in Phase II to estimate ${\bm {v}}_{1,n}$'s.

In Phase II, the signal received by antenna $m$, $m=1,\cdots,M$, at time instant $i$, $i=\tau_1+1\cdots,\tau_1+\tau_2$, is given as
\begin{align}\label{revS2m}
  \tilde{y}_{m,i}^{\rm II}=\sqrt{p}\sum_{n=1}^{N}\phi_{n,i}\beta_{m,n}{\bm v}_{1,n}^T{\bm x}_i.
\end{align}Define the overall signal received over Phase II as
\begin{align}
\tilde{\bm y}^{\rm II}=[\tilde{y}_{1,\tau_1+1}^{\rm II},\cdots,\tilde{y}_{M,\tau_1+1}^{\rm II}, \cdots, \tilde{y}_{1,\tau_1+\tau_2}^{\rm II},\cdots,\tilde{y}_{M,\tau_1+\tau_2}^{\rm II}]^T.
\end{align}Since both $\bar{\alpha}_{1,n}$'s in Phase I and $\tilde{\bm y}^{\rm II}$ in Phase II contain information about ${\bm {v}}_{1,n}$'s, define a vector ${\bm \delta}$ as
\begin{equation}\label{overallrevS2}
  {\bm \delta} = [ (\tilde{\bm y}^{\rm II})^T, \bar{{\bm \alpha}}_{1}^T]^T.
\end{equation}It can be shown that
\begin{equation}\label{overallS2}
  {\bm \delta}={\bm\Theta}\left[{\bm v}_{1,1}^T,\cdots,{\bm v}_{1,N}^T\right]^T,
\end{equation}
where $\bm\Theta\in\mathbb{C}^{(\tau_2M+N)\times KN}$ is given in \eqref{Theta} on the top of the next page. Our aim is to design the user pilot signals $x_{i,k}$'s and the IRS reflection coefficients $\phi_{n,i}$'s in Phase II such that ${\bm {v}}_{1,n}$'s can be estimated based on (\ref{overallS2}) with the minimum number of time instants. Mathematically, this requires to find the minimum $\tau_2$ such that rank$(\bm\Theta)=KN$, i.e., all the columns are linearly independent with each other.
\begin{figure*}[t]
\setlength{\columnsep}{0.1pt}
\begin{align}\label{Theta}
  {\bm \Theta} &=
  \sqrt{p}\left[
  \begin{array}{ccccccc}
    \phi_{1,\tau_1+1}{x}_{\tau_1+1,1}{\bm\beta}_{1} & \cdots & \phi_{1,\tau_1+1}{x}_{\tau_1+1,K}{\bm\beta}_{1} & \cdots & \phi_{N,\tau_1+1}{x}_{\tau_1+1,1}{\bm\beta}_{N} & \cdots & \phi_{N,\tau_1+1}{x}_{\tau_1+1,K}{\bm\beta}_{N}\\
    \vdots & \ddots & \vdots & \ddots & \vdots & \ddots & \vdots \\
    \phi_{1,\tau_1+\tau_2}{x}_{\tau_1+\tau_2,1}{\bm\beta}_{1} & \cdots & \phi_{1,\tau_1+\tau_2}{x}_{\tau_1+\tau_2,K}{\bm\beta}_{1} & \cdots & \phi_{N,\tau_1+\tau_2}{x}_{\tau_1+\tau_2,1}{\bm\beta}_{N} & \cdots & \phi_{N,\tau_1+\tau_2}{x}_{\tau_1+\tau_2,K}{\bm\beta}_{N} \\
    {\bm x} & \cdots & {\bm x} &  & {\bm 0} & \cdots & {\bm 0}  \\ {\bm 0} & \cdots & {\bm 0} & & {\bm 0} & \cdots & {\bm 0} \\
   \vdots & \ddots & \vdots & \ddots & \vdots & \ddots & \vdots \\
   {\bm 0} & \cdots & {\bm 0}& & {\bm x} & \cdots & {\bm x}
  \end{array}
  \right]
\end{align}
\hrulefill
\end{figure*}

\begin{theorem}\label{theorem1}
The minimum value of $\tau_2$ to guarantee a perfect estimation of ${\bm {v}}_{1,n}$'s based on (\ref{overallS2}) is
\begin{equation}\label{tau2}
  \tau_2^\ast=\max\left(K-1,\left\lceil \frac{(K-1)N}{M} \right\rceil\right).
\end{equation}
\end{theorem}

\begin{proof}
Please refer to Appendix \ref{appendix1}.
\end{proof}

According to (\ref{eqn:tau1}) and (\ref{tau2}), under our proposed protocol, in the case without BS noise, the minimum number of time instants to estimate all the reflected channels is still (\ref{minoverload}), which is the same as that achieved by the scheme proposed in \cite{b5}.

\section{LMMSE Channel Estimation Under the Case with BS Noise}\label{sec:LMMSE Channel Estimation Under the Case with BS Noise}

In this section, we will design the LMMSE channel estimators in the case with BS noise. Specifically, after removing the signal contributed by the direct user-BS links, the overall received signal at the BS over Phase I is
\begin{equation}\label{revS11_n}
    {\bm Y}^{\rm I}=\left[ \tilde{\bm y}_{1}^{\rm I},\cdots,\tilde{\bm y}_{M}^{\rm I}\right] = {\bm \Phi}^{\rm I}\left[ \bar{{\bm \alpha}}_{1},\cdots, \bar{{\bm \alpha}}_{M}\right]+{\bm Z}^{\rm I},
\end{equation}
where ${\bm Z}^{\rm I}=\left[ {\bm z}_{1},\cdots,{\bm z}_{\tau_1}\right]^T\sim \mathcal{CN}({\bm 0},\sigma^2 M {\bm I})$ is the receiver noise. Then, the LMMSE channel estimators can be designed as
\begin{align}\label{LMMSE_alpha}
 \hat{{\bm \alpha}}_m= & \left[ \hat{\alpha}_{m,1},\cdots,\hat{\alpha}_{m,N}\right]^T \nonumber \\ = & {\bm R}_{\alpha_m}({\bm \Phi}^{\rm I})^H\left( {\bm \Phi}^{\rm I}{\bm R}_{\alpha_m}({\bm \Phi}^{\rm I})^H+\sigma^2{\bm I} \right)^{-1}\tilde{\bm y}_{m}^{\rm I}, ~ \forall m,
\end{align}
where ${\bm R}_{\alpha_m}=\mathbb{E}\left[ \bar{{\bm \alpha}}_m \bar{{\bm \alpha}}_m^{H}\right]$ is the covariance matrix of $\bar{{\bm \alpha}}_{m}$. According to \eqref{alpha}, the correlation coefficients can be estimated as
\begin{equation}\label{LMMSE_beta}
  \hat{\beta}_{m,n}=\frac{\hat{\alpha}_{m,n}}{\hat{\alpha}_{1,n}}, ~ m=2,\cdots,M, ~ n=1,\cdots,N.
\end{equation}

In Phase II, the noisy version of ${\bm \delta}$ given in (\ref{overallS2}) based on the estimation of $\bar{{\bm \alpha}}_1$ given in \eqref{LMMSE_alpha} is
\begin{align}\label{revsigS2_n1}
    {\bm \delta}=&\left[ (\tilde{\bm y}^{\rm II})^T, \hat{{\bm \alpha}}_1^T\right]^T \nonumber \\ = & \hat{{\bm\Theta}}\left[{\bm v}_{1,1}^T,\cdots,{\bm v}_{1,N}^T\right]^T+({\bm\Theta}-\hat{{\bm\Theta}})\left[{\bm v}_{1,1}^T,\cdots,{\bm v}_{1,N}^T\right]^T\nonumber \\ & +\left[ ({\bm z}^{\rm II})^T, {\bm e}^T\right]^T,
\end{align}
where $\hat{{\bm\Theta}}$ is in the same form as ${\bm\Theta}$ given in (\ref{Theta}), but with $\beta_{m,n}$'s replaced by the estimation $\hat{\beta}_{m,n}$'s given in (\ref{LMMSE_beta}), ${\bm z}^{\rm II}=\left[ {\bm z}^{T}_{\tau_1+1},\cdots, {\bm z}^{T}_{\tau_1+\tau_2} \right]^T$ is the BS noise, and ${\bm e}=\hat{{\bm \alpha}}_1-\bar{{\bm \alpha}}_1$ denotes the estimation error of $\bar{{\bm \alpha}}_1$ in Phase I. The error propagated from Phase I to Phase II, i.e., ${\bm\Theta}-\hat{{\bm\Theta}}$ and ${\bm e}$, makes it hard to obtain the LMMSE estimators of ${\bm v}_{1,n}$'s. Similar to \cite{b5}, in the following, we ignore these errors such that
\begin{align}\label{revsigS2_n}
    {\bm \delta}=\hat{{\bm\Theta}}\left[{\bm v}_{1,1}^T,\cdots,{\bm v}_{1,N}^T\right]^T+\left[ ({\bm z}^{\rm II})^T, {\bm 0}^T\right]^T.
\end{align}In practice, we can increase the power and the number of pilots in Phase I such that (\ref{revsigS2_n}) is a good approximation of (\ref{revsigS2_n1}). Based on (\ref{revsigS2_n}), the LMMSE estimators of ${\bm v}_{1,n}$'s are
\begin{align}\label{LMMSE_v}
  \left[\hat{\bm v}_{1,1}^T,\cdots,\hat{\bm v}_{1,N}^T\right]^T={\bm R}_{v}\hat{{\bm \Theta}}^{H}\left( \hat{{\bm \Theta}}{\bm R}_{v}\hat{{\bm \Theta}}^{H}+ {\bm R}_z\right)^{-1}{\bm \delta},
\end{align}
where ${\bm R}_z={\rm diag}\{\sigma^2{\bm I}, {\bm 0}_{N}\}$ with ${\bm 0}_N$ being the $N$ by $N$ all-zero matrix, and ${\bm R}_{v}$ is the covariance matrix of $\left[{\bm v}_{1,1}^T,\cdots,{\bm v}_{1,N}^T\right]^T$.

\begin{remark}
  In the LMMSE estimators in \eqref{LMMSE_alpha} and \eqref{LMMSE_v}, we need the knowledge about ${\bm R}_{\alpha_m}$ and ${\bm R}_{v}$ for ${\bm \alpha}_m$ and $\left[{\bm v}_{1,1}^T,\cdots,{\bm v}_{1,N}^T\right]^T$, respectively. This can be obtained via the distribution of $r_{n,m}$ and $t_{k,n}, \forall n,m,k$. For example, let us consider the channel model in \cite{b5}, where
  \begin{equation}
    r_{n,m}=\sum_{i=1}^{N}\tilde{r}_{i,m}({\bm C}^{\rm I})^{\frac{1}{2}}_{i,n}, \forall m,n,
  \end{equation}
  with $\tilde{r}_{i,m}\sim\mathcal{N}(0,\beta^{\rm  BI})$ is the i.i.d CSCG component and $({\bm C}^{\rm I})^{\frac{1}{2}}$ denotes the IRS transmit correlation matrix, and
  \begin{equation}
    t_{k,n}=\sum_{i=1}^{N}({\bm C}^{\rm I}_{k})^{\frac{1}{2}}_{i,n}\tilde{t}_{k,i}, \forall k,n,
  \end{equation}
   with $\tilde{t}_{k,i}\sim\mathcal{N}(0,\beta_k^{\rm IU})$ is also the i.i.d CSCG component and $({\bm C}^{\rm I}_{k})^{\frac{1}{2}}$ denotes the IRS receive correlation matrix for user $k$. In this case, the $c$th row and $l$th column entry of ${\bm R}_{\alpha_m}$ is given by
   \begin{equation}\label{R_alpham}
     {\bm R}_{\alpha_m,c,l}=p\beta^{\rm BI}({\bm C}^{\rm I})^T_{c,l}{\bm x}^T{\bm R}_{t_{cl}}{\bm x},~c,l=1,\cdots,N,
   \end{equation}
    where the $k_1$th row and $k_2$th column entry of ${\bm R}_{t_{cl}}$ is calculated as
    \begin{align}\label{R_t_cl_2}
        &{\bm R}_{t_{cl},k_1,k_2}= \notag\\
        &\begin{cases}
        $$\beta_{k}^{\rm IU}({\bm C}^{\rm I}_{k})^{\frac{1}{2}}_{c}\left( ({\bm C}^{\rm I}_{k})^{\frac{1}{2}}_{l}\right)^{H} $$, & \mbox{if $k_1=k_2=k$}, \\
        $0$, & \mbox{otherwise}.
        \end{cases}
    \end{align}
    $({\bm C}^{\rm I}_{k_1})^{\frac{1}{2}}_{c}$ and $({\bm C}^{\rm I}_{k_2})^{\frac{1}{2}}_{l}$ denotes the $c$th row of $({\bm C}^{\rm I}_{k_1})^{\frac{1}{2}}$ and $l$th row of $({\bm C}^{\rm I}_{k_2})^{\frac{1}{2}}$, respectively. And the sub-block of $(c-1)K+1$ to $(c K)$th rows and $(l-1)K+1$ to $(lK)$th columns of ${\bm R}_{v}$, i.e., ${\bm R}_{v,c,l}$ is given by
    \begin{equation}\label{R_vn}
        {\bm R}_{v,c,l}=\beta^{\rm BI}({\bm C}^{\rm I})^T_{c,l}{\bm R}_{t_{cl}},~~c,l=1,\cdots,N. \notag
    \end{equation}

\end{remark}

\section{Numerical Results}\label{sec:Numerical Results}
In this section, we present numerical results to show the gain of our scheme over that proposed in \cite{b5} in the case with BS noise. We assume that the numbers of IRS elements, BS antennas, and users are $N=32$, $M=32$, and $K=8$, respectively. All the channels are generated in the same way as those used in \cite{b5}. The identical transmit power of all the users is $23$ dBm. The channel bandwidth is assumed to be $1$ MHz, and the power spectrum density of the noise at the BS is $-169$ dBm/Hz. We use the normalized MSE (NMSE) as the performance indicator. For any vector ${\bm a}$, if its estimation is $\hat{\bm a}$, then the NMSE is defined as
\begin{equation}\label{NMSE}
  {\rm NMSE}=\frac{\mathbb{E}\left[\parallel\hat{\bm a}-{\bm a}\parallel^{2}\right]}{\mathbb{E}\left[\parallel{\bm a}\parallel^{2}\right]}.
\end{equation}In our scheme, ${\bm a}=[{\bm v}_1^T,\cdots,{\bm v}_M^T]^T$, and in the benchmark scheme proposed in \cite{b5}, ${\bm a}=[{\bm g}_1^T,\cdots,{\bm g}_K^T]^T$. In the following, we will compare the NMSE performance between our proposed scheme and the benchmark scheme in \cite{b5}.\footnote{In \cite{b5}, there are three phases in the protocol, because the direct user-BS channels are estimated in the first phase. To make it consistent with our work, in the following, we will ignore their phase to estimate the direct channels, and call their phase to estimate ${\bm g}_{1,n}$'s as Phase I, and their phase to estimate $\lambda_{k,n}$'s as Phase II.}

Fig. \ref{Phase I} shows the NMSE performance over the pilot sequence length, which ranges from $40$ to $100$ symbols. According to (\ref{eqn:tau1}) and (\ref{tau2}), the minimum numbers of pilot symbols required by Phases I and II are $32$ and $7$, respectively. Here, we assume that all the extra pilot symbols are allocated to Phase I. It is observed that with imperfect estimation in Phase I, the overall NMSE achieved by our scheme after two phases is smaller than that achieved by the scheme in \cite{b5}. To analyze the reason, we also provide the NMSE performance when the estimation of Phase I is assumed to be perfect. In this case, it is observed that the NMSE achieved by \cite{b5} is even smaller than that of our scheme. This indicates that  our gain in terms of overall NMSE over \cite{b5} comes from Phase I, where the transmit power of all the users is utilized under our scheme, but the power of the other $K-1=7$ users is wasted under the scheme in \cite{b5}. In Fig. \ref{Phase I}, we also plot the NMSE of the scheme in \cite{b5} when the typical user can transmit with a power $Kp=8p$ in Phase I. In this case, it is observed that indeed, if the total transmit power is the same in Phase I, the NMSE performance will be very close in both schemes. Similar observations can be also found in Fig. \ref{Phase II}, where all the extra pilot symbols are allocated to Phase II.

\begin{figure}[t]
	\centering
	\includegraphics[width=9cm]{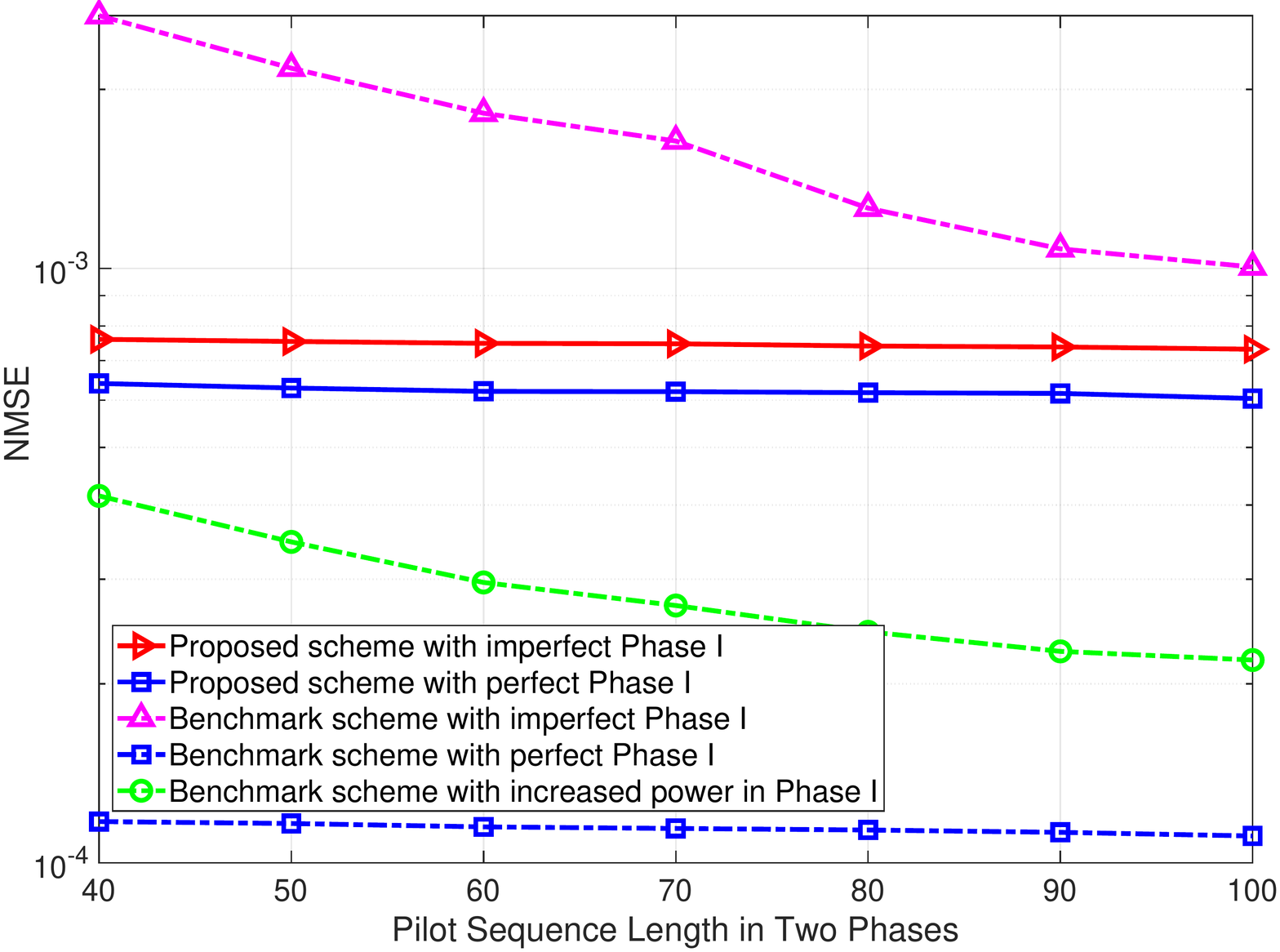}\vspace{-10pt}
	\caption{Performance comparison when extra pilots are allocated to Phase I.} \vspace{-10pt} \label{Phase I}
\end{figure}

\begin{figure}[t]
	\centering
	\includegraphics[width=9cm]{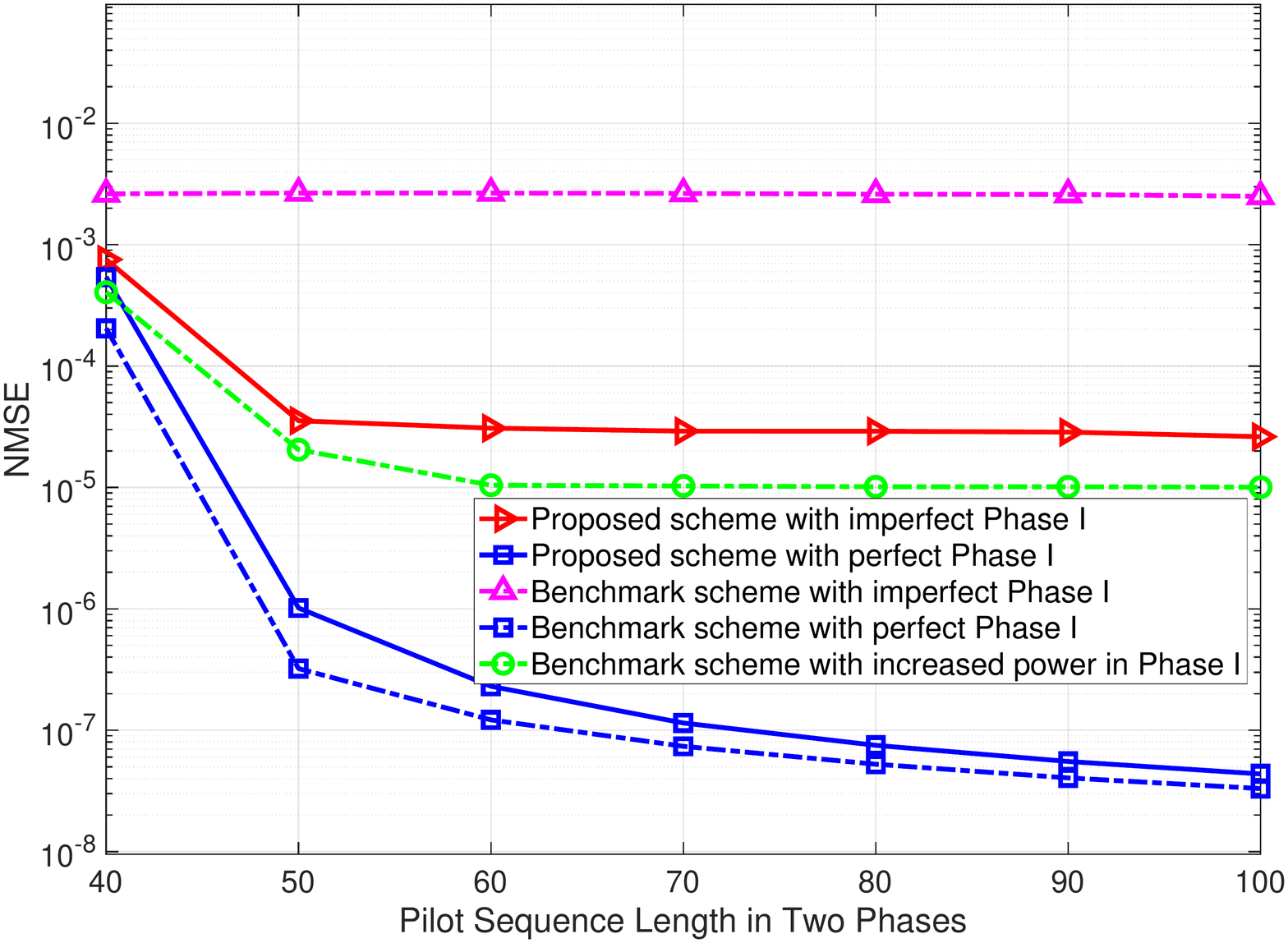}\vspace{-10pt}
	\caption{Performance comparison when extra pilots are allocated to Phase II.}\label{Phase II} \vspace{-15pt}
\end{figure}

\begin{figure}[t]
  \centering
  \includegraphics[width=8.5cm]{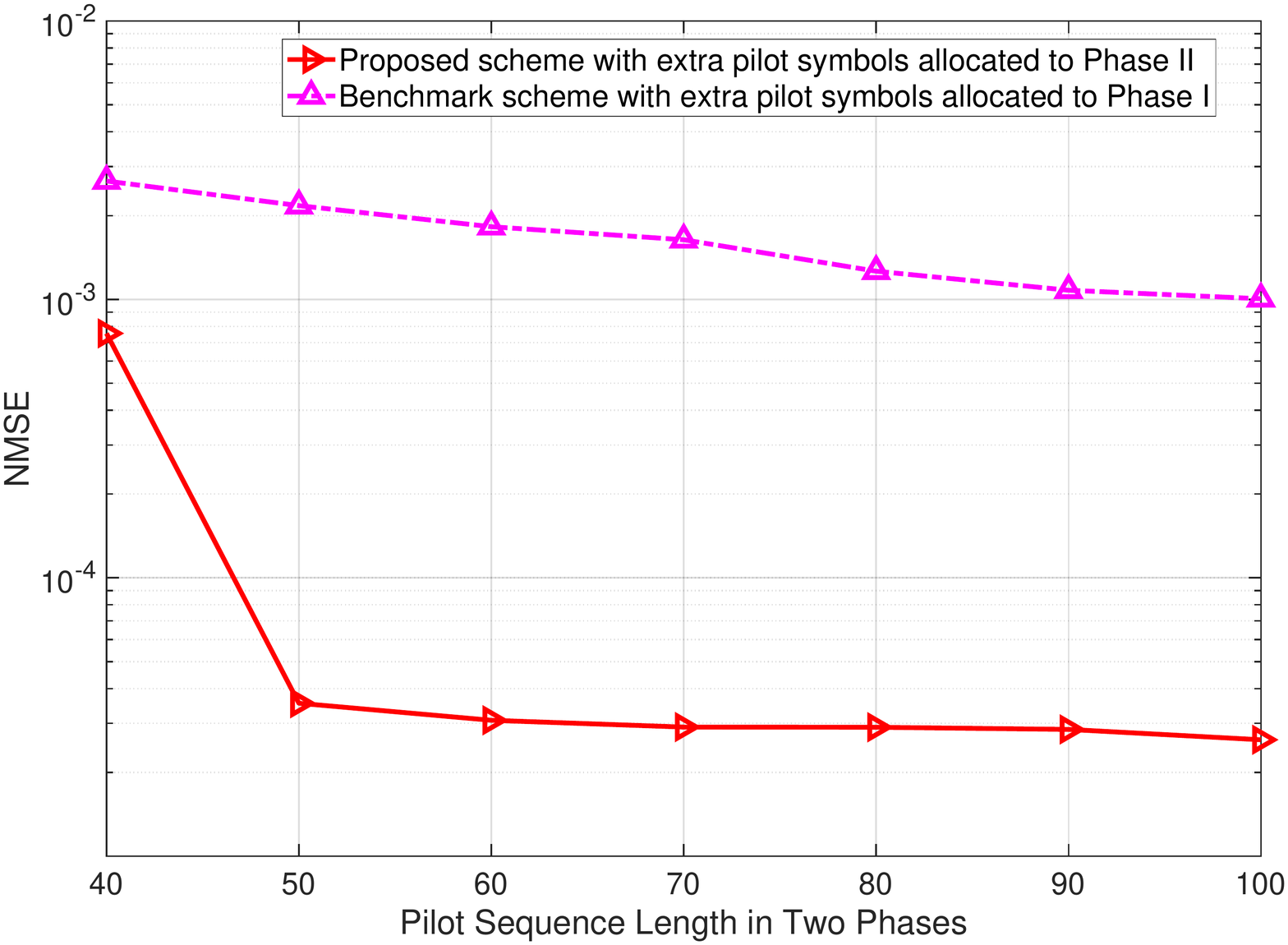}\vspace{-10pt}
  \caption{Performance comparison between the best strategies.}\label{best_strategy} \vspace{-15pt}
\end{figure}

Finally, notice from Fig. \ref{Phase I} and Fig. \ref{Phase II} that for our proposed scheme, it is optimal to allocate all extra pilots to Phase II, as Phase I is already very good; while for the scheme in \cite{b5}, it is optimal to allocate all extra pilots to Phase I to improve its performance. In Fig. \ref{best_strategy}, we compare the NMSE of the two schemes with their optimal pilot allocation strategies as described above, over the total pilot sequence length in the two phases. It is observed that under the best strategies for both schemes, our proposed scheme significantly outperforms the scheme in \cite{b5}, due to the full exploitation of all user energy.

\section{Conclusions}\label{sec:Conclusions}
In this paper, we revealed a new correlation among the cascaded user-IRS-BS channels in the uplink communication of an IRS-assisted multi-user network. Specifically, the cascaded user-IRS-BS channel vector for a BS antenna is a scaled version of that for any other antenna because of the common user-IRS channel. Based on this property, we proposed a novel two-phase channel estimation protocol, where the scaling coefficients of all the other antennas' channels are estimated in Phase I, and the cascaded channel of the typical antenna is estimated in Phase II. In the case without BS noise, the minimum number of time instants required for perfect channel estimation was characterized, which was shown to be the same as that achieved by the strategy proposed in \cite{b5}. In the case with BS noise, the LMMSE channel estimators were proposed for both Phase I and Phase II. Via simulation, it was shown that the MSE achieved by our proposed scheme is much smaller than that achieved by the scheme proposed in \cite{b5} due to the better utilization of user energy.

\appendix
\subsection{Proof of Theorem \ref{theorem1}}\label{appendix1}
We prove the theorem in the following two cases: 1) $M\geq N$ and 2) $M<N$. In the case of $M\geq N$, we first prove that there exist a unique solution to \eqref{overallS2} only if $\tau_2\geq K-1$. Define
\begin{equation}\label{eta}
  \eta_{n,i}=\sqrt{p}{\bm v}_{1,n}^{T}{\bm x}_i, n=1,\cdots, N, i=\tau_1+1,\cdots,\tau_1+\tau_2.
\end{equation}
Then it can be shown that the received signal by antenna $m$ in \eqref{revS2m} can be expressed as
\begin{equation}\label{revS2m'}
  \tilde{y}_{m,i}^{\rm II}=\sum_{n=1}^{N}\phi_{n,i}\beta_{m,n}\eta_{n,i}, i=\tau_1+1,\cdots,\tau_1+\tau_2.
\end{equation}
With $\beta_{m,n}$'s estimated in Phase I, for each time instant $i$, there exist $N$ variables $\eta_{n,i}$'s and $M$ linear equations as given in \eqref{revS2m'}. As a result, in the case of $M\geq N$, $\eta_{n,i}$'s can be perfectly estimated. Then, with the knowledge of $\eta_{n,i}$'s and $\bar{\alpha}_{1,n}$'s, we can estimate ${\bm v}_{1,n}$'s from the following equations
\begin{align}\label{eta&alpha}
  &\eta_{n,i}=\sqrt{p}{\bm v}_{1,n}^{T}{\bm x}_{i}, \bar{\alpha}_{1,n}=\sqrt{p}{\bm v}_{1,n}^{T}{\bm x},\\
  &n=1,\cdots, N, i=\tau_1+1,\cdots,\tau_1+\tau_2,\notag
\end{align}
which characterizes a linear system with $KN$ variables and $(\tau_2 N+N)$ equations. Therefore, a unique solution to \eqref{eta&alpha} exists only when the number of equations is no smaller than the number of variables, i.e. $\tau_2\geq K-1$.

Next, we show that if $\tau_2 =K-1$, there always exists a unique solution to \eqref{overallS2} in the case of $M\geq N$. Specifically, we set $\phi_{n,i}$'s in \eqref{revS2m'} as one, $\bm x$ as an all one vector and the stacked vectors $[{\bm x}_{\tau_1+1},\cdots,{\bm x}_{\tau_1+K-1}]$ as the $2$ to $K$ columns of a $K \times K$ DFT matrix. Since ${\bm\beta}_n$'s are linearly independent with each other with probability one , $\eta_{n,i}$'s can be perfectly estimated as
\begin{align}\label{es_eta}
  &[\eta_{1,i},\cdots,\eta_{N,i}]^T=[\tilde{y}^{\rm II}_{1,i},\cdots,\tilde{y}^{\rm II}_{M,i}]^T{\bm\beta}^\dagger,\\
  &i=\tau_1+1,\cdots,\tau_1+K-1,\notag
\end{align}
where
\begin{equation}\label{Beta}
  {\bm\beta}=\left[
  \begin{array}{ccc}
    \beta_{1,1} & \cdots & \beta_{1,N} \\
    \vdots & \ddots & \vdots \\
    \beta_{M,1} & \cdots & \beta_{M,N}
  \end{array}\right],
\end{equation}
and for any matrix $\bm A\in\mathbb{C}^{s\times t}$ with $s\geq t$, ${\bm A}^{\dagger}=({\bm A}^H{\bm A})^{-1}{\bm A}^H$ denotes its pseudo-inverse matrix. Then, as a result, if $\tau_2=K-1$, there exists a unique solution to \eqref{eta&alpha}, equivalently to  \eqref{overallS2} given as follows
\begin{align}\label{es_v1_M>N}
{\bm v}_{1,n}^T&=\frac{1}{\sqrt{p}}[\bar{\alpha}_{1,n},\eta_{n,\tau_1+1},\cdots,\eta_{n,\tau_1+K-1}] \notag\\
&[{\bm x},{\bm x}_{\tau_1+1},\cdots,{\bm x}_{\tau_1+K-1}]^{\dagger}, n=1,\cdots,N.
\end{align}

Next, consider the case of $M<N$. Since the number of variables and equations in \eqref{overallS2} are $KN$ and $\tau_2M+N$, respectively, there exists a unique solution to \eqref{overallS2} only if the number of equation is no smaller than that of variables, i.e., $\tau_2\geq\lceil\frac{(K-1)N}{M}\rceil$.

Next, we show that when $\tau_2=\lceil\frac{(K-1)N}{M}\rceil$, there always exists a solution to \eqref{overallS2} in the case of $M<N$. Specifically, we first set the pilot signal in Phase I as
\begin{equation}
  {\bm x}=[1,\cdots,1]^T.
\end{equation}
Then, in Phase II, user $K$ does not transmit the pilot signal, i.e., ${\bm x}_{i,K}=0, i=\tau_1+1, \cdots, \tau_1+\tau_2$, while the pilot signals of user $1$ to user $K-1$, i.e., ${\bm x}_{i,k}, i=\tau_1+1,\cdots,\tau_1+\tau_2,k=1,\cdots,K-1$, as well as the IRS reflecting coefficients, i.e., $\phi_{n,i}, n=1,\cdots,N, i=\tau_1+1,\cdots,\tau_1+\tau_2$, are set as Theorem $2$ in \cite{b5}. Then, we construct a new matrix $\tilde{\bm \Theta}\in\mathbb{C}^{(\tau_2M+N)\times KN}$ by putting the $[(n-1)K+k]$th column of $\bm \Theta$ in \eqref{Theta} into the $[(k-1)N+n]$th column of $\tilde{\bm \Theta}$, $\forall n,k$. Since changing the order of columns of a matrix does not change its rank, i.e., ${\rm rank}(\tilde{\bm \Theta})={\rm rank}(\bm \Theta)$, in the following, we show that under the above construction, we have ${\rm rank}(\tilde{\bm \Theta})=KN$ when $\tau_2=\lceil\frac{(K-1)N}{M}\rceil$. Specifically, $\tilde{\bm \Theta}$ can be re-expressed as follows:
\begin{equation}\label{Theta_tilde_block}
  \tilde{\bm \Theta}=\left[
  \begin{array}{cc}
    \tilde{\bm\Theta}_{s} & {\bm O}_{\tau_2M \times N}\\
    \{{\bm I}_N\}_{K-1} & {\bm I}_N
  \end{array}\right],
\end{equation}
where $\tilde{\bm\Theta}_{s}$ is the first $\tau_2M$ rows and first $(K-1)N$ columns of $\tilde{\bm\Theta}$, ${\bm I}_N$ is the identity matrix of dimension $N$, and $\{{\bm I}_N\}_{K-1}=[{\bm I}_N,\cdots,{\bm I}_N]\in\mathbb{C}^{N\times (K-1)N}$. According to Theorem $2$ in \cite{b5}, ${\rm rank}(\tilde{\bm \Theta}_s)=(K-1)N$ when $\tau_2=\lceil\frac{(K-1)N}{M}\rceil$. Next, we derive the rank of $\tilde{\bm\Theta}$. It is observed from \eqref{Theta_tilde_block} that each of the first $\tau_2 M$ rows of $\tilde{\bm \Theta}$, whose last $N$ elements are all zero, is linearly independent of the last $N$ rows of $\tilde{\bm\Theta}$, i.e., $\{{\bm I}_N\}_{K}$. In other words, the row space defined by the first $\tau_2 M$ rows in $\tilde{\bm\Theta}$ does not intersect with that defined by the last $N$ rows in $\tilde{\bm\Theta}$. In this case, ${\rm rank}(\tilde{\bm \Theta})={\rm rank}([\tilde{\bm \Theta}_s~~{\bm O}_{\tau_2M\times N}])+{\rm rank}(\{{\bm I}_N\}_{K})=KN$ \cite{Rank}. Therefore, for the case of $M<N$, when $\tau_2=\lceil\frac{(K-1)N}{M}\rceil$, there exists a unique solution to \eqref{overallS2} given by
\begin{equation}\label{es_v1_M<N}
  \left[{\bm v}_{1,1}^{T},\cdots,{\bm v}_{1,N}^{T}\right]^{T}={\bm\Theta}^{\dagger}{\bm \delta}.
\end{equation}

Theorem $1$ is thus proved.

\bibliographystyle{IEEEtran}
\bibliography{IRS_channel_estimation}

\end{document}